\documentclass{article}

\newcommand*{\APPENDIX}{}

\usepackage[utf8]{inputenc}
\usepackage[T1]{fontenc}

\usepackage{basel}

\title{Special-case Algorithms for Blackbox Radical Membership, Nullstellensatz and Transcendence Degree}

\author{
Abhibhav Garg \\ CSE, Indian Institute of Technology Kanpur \\ \texttt{abhibhav@cse.iitk.ac.in}
\and
Nitin Saxena \\ CSE, Indian Institute of Technology Kanpur \\ \texttt{nitin@cse.iitk.ac.in}
}
\date{}

\begin{document}
\maketitle
\begin{abstract}
Radical membership testing, resp.~its special case of Hilbert's Nullstellensatz (HN), is a fundamental computational algebra problem.
It is NP-hard; and has a famous PSPACE algorithm due to {\em effective} Nullstellensatz bounds.
We identify a useful case of these problems where practical algorithms, \& improved bounds, could be given--- When transcendence degree $r$ of the input polynomials is smaller than the number of variables $n$.
If $d$ is the degree bound on the input polynomials, then we solve radical membership (even if input polynomials are {\em blackboxes}) in around $d^r$ time.
The prior best was $> d^n$ time (always, $d^n\ge d^r$). Also, we significantly improve effective Nullstellensatz degree-bound, when $r\ll n$.

Structurally, our proof shows that these problems reduce to the case of $r+1$ polynomials of transcendence degree $\ge r$.
This input instance (corresponding to none or a {\em unique annihilator}) is at the core of HN's hardness.
Our proof methods invoke basic algebraic-geometry.
\end{abstract}

\section{Introduction}

Given a set of polynomials $f_{1}, \dots, f_{n}$, there is a natural certificate for the existence of a common root, namely the root itself.
Hilbert's Nullstellensatz \cite{rabinowitsch1930hilbertschen, zariski1947new, krull1950jacobsonsches} states that there is also a natural certificate for the nonexistence of a common root, when the underlying field is {\em algebraically closed}.
Formally, the theorem states that the polynomials have no common root if and only if there exist polynomials $g_{1}, \dots, g_{n}$ such that $1 = \sum f_{i} g_{i}$.
We refer to the latter type of certificate as a {\em Nullstellensatz certificate}.
These certificates are not polynomial sized: every common root can have exponential bit complexity, and every set of witness polynomials $g_{i}$ can have exponential degrees.
This problem is naturally of computational interest, since the generality of the statement affords reductions from many problems of interest.
{\em Effective} versions of the Nullstellensatz have been extensively studied \cite{jelonek, krick2001sharp, krick1999arithmetic, sombrasparse, sombrahilbert, berensteinstruppasurvey, kollar, brownawell}, and they allow the decision problem of existence of common roots (called HN) to be solved in polynomial space.
Koiran \cite{koiranph} proved that under generalized Riemann Hypothesis, HN can be solved in AM \cite[Ch.8]{arora2009computational}, for fields of characteristic zero.

In this work, we relate the complexity of HN to the transcendence degree of the input polynomials.
The {\em transcendence degree} of polynomials $f_{1}, \dots, f_{m}$ is defined as the size of any maximal subset of the polynomials that are algebraically independent.
This notion is well defined since algebraic independence satisfies matroid properties \cite{oxley}.
We show that HN can be solved in time single-exponential in transcendence degree
This can be seen as a generalization of the fact that HN can be solved in time exponential in the number of polynomials (or variables) in the system. We state our result in terms of the question of {\em radical membership}: $f_0 \in? \sqrt{\langle f_{1}, \dots, f_{m} \rangle}$. Note that the standard algorithms for both ideal membership \cite{hermann} and radical computation \cite{laplagne2006algorithm} are far slower than ours. 

\emph{
  Given a set of polynomials $f_{1}, \dots, f_{m}$ with transcendence degree at most $r$, as blackboxes, we can perform radical membership tests for the ideal generated by $f_{1}, \dots, f_{m}$ in time polynomial in $d^r,m,n$, where $d$ is the degree-bound on the polynomials and $n$ is the number of variables.
}

We also relate the transcendence degree of the input polynomials to the degrees of the Nullstellensatz certificates, that is the degrees of $g_{i}$ in $\sum f_{i} g_{i} = 1$; improving the best bounds by \cite{jelonek}.

\emph{
  Given a set of polynomials $f_{1},\dots, f_{m}$ with transcendence degree $r$ and without any common roots, there exist polynomials $g_{i}$ of degree at most $d^{r+1}$ such that $\sum f_{i} g_{i} = 1$.
}

We also give an output-sensitive algorithm to compute the transcendence degree of polynomials.
Slightly more formally, we show:

\emph{
  Given a set of polynomials $f_{1}, \dots, f_{m}$, we can compute their transcendence degree in time polynomial in $d^{r}$ and $m,n$.
}


\vspace{-1mm}
\subsection{Previously known results}
All three of the problems stated above have been extensively studied.
We therefore only list some of the previously known results, and direct readers to the surveys \cite{mayrsurvey, berensteinstruppasurvey}.

{\bf Nullstellensatz.} The decidability of the ideal membership problem was established by Hermann \cite{hermann} when she proved a doubly-exponential bound on witnesses to ideal membership.
A lower bound of the same complexity by Mayr and Meyer \cite{mayr89, mayrmeyer82} showed that this problem is EXPSPACE complete.
A number of different algorithms were developed for operations on ideals, most prominently the method of Gr\"obner basis \cite{gbthesis}.
The proof of single-exponential bounds for the Nullstellensatz (discussed below) allowed special cases of the ideal membership problem, such as the case of unmixed and zero dimensional ideals to be solved in single-exponential time \cite{dfgsunmixed}.
It also allowed the general Nullstellensatz problem to be solved in PSPACE.
Giusti and Heintz \cite{ghf} proved that the dimension of a variety can be computed by a randomized algorithm in single-exponential time, with the exponent being linear in $n$, which gives an algorithm of the same complexity for HN (by testing if the dimension is $-1$).
All of the above results are independent of the underlying field characteristic.
In 1996, Koiran \cite{koiranph} gave an AM protocol (conditioned on GRH) for the Nullstellensatz problem, when the underlying field is $\bC$ and the polynomials have integer coefficients.
His method is completely different from the previous methods (of using the effective Nullstellensatz to reduce the system to a linear one).
The positive characteristic case is an open problem, and the best known complexity remains PSPACE.

{\bf Effective Nullstellensatz.} The projective version of the effective Nullstellensatz follows from the fundamental theorem of elimination theory \cite{lazard1977}.
An affine version was first proved by Brownawell \cite{brownawell} in characteristic $0$ using analytic methods.
It was later improved by Kollár \cite{kollar} who used local cohomology to improve the bounds and remove the condition on the characteristic.
A more elementary proof that used bounds on the Hilbert function was given by Sombra \cite{sombrahilbert}, who also gave improved bounds based on some geometric properties of related varieties \cite{sombrasparse}.
An even more elementary and significantly shorter proof was given by Jelonek \cite{jelonek}, who obtained improved bounds when the number of polynomials is lesser than the number of variables.

{\bf Transcendence degree.} Algebraic independence was studied in computer science by \cite{dgw} in their study of explicit extractors.
They proved that the rank of the Jacobian matrix is the same as the transcendence degree for fields of characteristic zero (or large enough) which gives an efficient randomized method for computing the transcendence degree.
The problem was studied further in \cite{kayal}, where the condition on the characteristic for the above algorithm was relaxed, and some hardness results were established.
\cite{gss} showed that the problem is in coAM $\cap$ AM, making it unlikely to be NP-hard, and conjecturing that the problem is in coRP for all characteristics.
Algorithmically, the best known method for computing the transcendence degree in fields of positive characteristic still has PSPACE complexity, by using the bounds of Perron \cite{perron1951algebra, ploski} to reduce the problem to solving an exponential sized linear system.
This method takes time polynomial in $d^{r^{2}}$ using the methods of \cite{csanky}.
We refer the reader to the thesis \cite{amitthesis} for an exhaustive survey of related results; and applications in \cite{ASSS11,PSS16}.

Certain radical membership methods were developed by Gupta \cite{ankitgupta} in his work on deterministic polynomial identity testing algorithms for heavily restricted depth-four circuits.
The focus there however was on a {\em deterministic algorithm} for the above problem.
Further, he restricts his attention to systems where the underlying field is $\bC$.

\vspace{-1mm}
\subsection{Our results}
Our algorithms will be Monte Carlo algorithms.
We assume that our base field $k$ is algebraically closed, but our algorithms only use operations in the field in which the coefficients of the inputs lie, which we denote by $k_{i}$.
For example, $k_{i}$ might be $\bF_{p}$, and $k$ would then be $\overline{\bF_{p}}$.
By time complexity we mean operations in $k_{i}$, where operations include arithmetic operations, finding roots, and computing GCD of polynomials.
Our results are valid for any field where the above procedures are efficient, for example finite fields.

We relate the complexity of radical membership, and the degree bounds in effective Nullstellensatz, to the transcendence degree of the input set of polynomials.
We do this by showing that given a system of polynomials, we can reduce both the number of variables and the number of polynomials to one more than the transcendence degree, while preserving the existence (resp.~non-existence) of common roots.
In particular, when the transcendence degree of the input polynomials is constant, we get efficient algorithms for these problems.
\begin{theorem}[Radical membership]
  \label{thm:main}
  Suppose $f_{1}, \dots, f_{m}$ and $g$ are polynomials, in variables $x_{1}, \dots, x_{n}$, of degrees $d_{1}, \dots, d_{m}$ and $d_{g}$ respectively, given as blackboxes.
  Suppose that $\trdeg{f_{1}, \dots, f_{m}} \leq r$.
  Define $d:= \max (\max_{i} d_{i}, d_{g})$.
  
  Then, testing if $g$ belongs to the radical of the ideal generated by $f_{1}, \dots, f_{m}$ can be done in time polynomial in $n, m$ and $d^{r}$, with randomness.
\end{theorem}
\noindent\textbf{Remarks:}

\textbf{(1)}
The tr.deg $r$ can be much smaller than $n$, and this improves the complexity significantly to $d^{r}$ from the prior $d^{n}$ \cite{lakshmanlazard}.
On the other hand, the usual reduction from SAT to HN results in a set of polynomials with transcendence degree $n$, due to the presence of polynomials $x_{i}^{2} - x_{i}$ (that enforce the binary 0/1 values).

\textbf{(2)}
We also show that the tr.deg itself can be computed in time $d^{r}$, independent of the characteristic (\cref{thm:trdegcompute}).
In the above statement therefore, we can always pick $r = \trdeg{\vf}$, and we can assume that $r$ is not part of the input.

\textbf{(3)}
The transcendence degree is upper bounded by the number of polynomials, and therefore we generalize the case of few polynomials.
It is surprising if one contrasts this case with that of {\em ideal} membership--- where the instance with three polynomials (i.e. transcendence degree = 3) is as {\em hard} as the general instance making it EXPSPACE-complete.
\footnote{
  Suppose $g \in \ideal{f_{1}, \dots, f_{m}}$ is an instance of ideal membership.
  This is equivalent to $z_{1}^{m}z_{2}^{m}g \in \ideal{z_{1}^{m+1}, z_{2}^{m+1}, \sum_{i}f_{i}z_{i}^{i}z_{2}^{m-i}}$.
  Here, $z_{1}, z_{2}$ are fresh variables.
  This reduces the general instance of ideal membership to an instance where the ideal is generated by $3$ elements.
  This transformation is from \cite{ramprasadpc}.
}

\smallskip
Next, we show that taking constant-free random linear combinations preserves the zeroset of the polynomials, if the number of linear combinations is at least one more than the transcendence degree.
This allows us to get bounds on the Nullstellensatz certificates that depend on the transcendence degree.
\begin{theorem}[Effective Nullstellensatz]
  \label{thm:main2}
  Suppose $f_{1}, \dots, f_{m}$ are polynomials in $x_{1}, \dots, x_{n}$, of degrees $d_{1} \geq \cdots \geq d_{m}$ respectively, with an empty zeroset.
  Suppose further that $\trdeg{f_{1}, \dots, f_{m}} = r$.
  
  Then, there exist polynomials $h_{i}$ such that $\deg f_{i} h_{i} \leq \prod_{i=1}^{r+1} d_{i}$ that satisfy $\sum f_{i} h_{i} = 1$.
\end{theorem}
\noindent\textbf{Remark:}
The prior best degree-bound for the case of `small' transcendence degree is $\prod_{i=1}^{m} d_{i}$
\cite{jelonek}.
Our bound is significantly better when the transcendence degree $r$ is `smaller' than the number of polynomials $m$.

\smallskip
Finally, as stated before, we show that the transcendence degree of a given system of polynomials can be computed in time polynomial in $d^{r}$ (and $m,n$), where $d$ is the maximum degree of the input polynomials, and $r$ is their transcendence degree.
The algorithm is output-sensitive in the sense that the time-complexity depends on the output number $r$.
\begin{theorem}[Transcendence degree]
  \label{thm:trdegcompute}
  Given as input polynomials $f_{1}, \dots, f_{m}$, in variables $x_{1}, \dots, x_{n}$, of degrees at most $d$, we can compute the transcendence degree $r$ of the polynomials in time polynomial in $d^{r}, n,m$.
\end{theorem}

\noindent\textbf{Remark:}
In the case when the characteristic of the field is greater than $d^{r}$, there is a much more efficient (namely, randomized polynomial time) algorithm using the Jacobian criterion \cite{bms}.
The algorithm presented here is useful when the characteristic is `small'; whereas the previous best known time-complexity was $> d^{r^{2}}$ if one directly implements the PSPACE algorithm. Eg.~for $d=O(1)$ and $r=O(\log n)$ our complexity is polynomial-time unlike the prior known algorithms.

\noindent\textbf{A motivating example} 
 where our results are better than the known results is when the input blackboxes are implicitly of the form $f_{i}(h_{1}, \dots, h_{r})$, $i\in[m]$, for $r\ll n$, where each $h_{i}$ is an $n$-variate polynomial, and $m = n+1$. Here, $f_{i}$'s have transcendence degree $r$. Thus, our algorithms take time $d^{r}$; significantly less than $d^{n}$.

\subsection{Proof ideas}
\textit{Pf.~idea \cref{thm:main}:}
We first use the Rabinowitsch trick to reduce to HN: the case $g = 1$.
Next, we perform a random linear variable-reduction. We show that replacing each $x_{i}$ with a linear combination of $r$ new variables $z_{j}$ preserves the existence of roots.
This is done by using the fact that a general linear hyperplane intersects a variety properly (Lemma \ref{lem:intdim}).
Once we are able to reduce the variables, we can interpolate to get dense representation of our polynomials, and invoke existing results about testing nonemptiness of varieties (\cref{thm:ll}).

\smallskip
\textit{Pf.~idea \cref{thm:main2}:}
For the second theorem, we show that random linear combinations of the input polynomials, as long as we take at least $r+1$ many of them, preserve the zeroset.
For this, we study the image of the polynomial map defined by the polynomials.
We again use the theorem regarding the hyperplane intersection (Lemma \ref{lem:intdim}).
In order to get the degree bounds, we must allow these hyperplanes to depend on fewer variables, and allow their equations to be constant free.
Once this is proved, we can use a bound (\cref{thm:jelonekeffective}) on the Nullstellensatz certificates for the new polynomials (which is better since the polynomials are fewer in number) to obtain a bound for the original polynomials.

\smallskip
\textit{Pf.~idea \cref{thm:trdegcompute}:}
The image of the polynomial map defined by the polynomials is such that the general fibre has codimension equal to the transcendence degree.
We first show that a random point, with coordinates from a subset which is not `too large', satisfies this property.
In order to efficiently compute the dimension of this fibre, we take intersections with hyperplanes; and apply Lemma \ref{lem:intdim} and \cref{thm:ll}.

\section{Notation and preliminaries}

\subsection{Notation}
We reserve $n$ for the number of variables ($x_1,\ldots,x_n$), $m$ for the number of polynomials ($f_{1}, \dots, f_{m}$) in our inputs.
The polynomials have total degrees $d_{1}, \dots, d_{m}$.
We assume that the polynomials are labeled such that $d_{1} \geq d_{2} \geq \cdots \geq d_{m}$.

We use boldface to denote sequence of objects, when the indexing set is clear; for example, $\vx$ denotes $x_{1}, \dots, x_{n}$ and $\vf$ denotes $f_{1}, \dots, f_{m}$.
The point $(0, \dots, 0)$ will be represented by $\vzero$.
We use $k$ to denote the underlying field which we assume is algebraically closed, and $k_{i}$ to denote the field in which the coefficients of the inputs lie.
We use $\bA^{n}$ to denote the $n$ dimensional affine space over $k$.
Given a variety $X$, we use $k\bs{X}$ to denote its coordinate ring, and when $X$ is irreducible we use $k(X)$ to denote its function field.
We use $\bA^{n}$ and $\bP^{n}$ to denote the $n$ dimensional affine and projective spaces respectively, and $\bP_{\infty}^{n}$ to denote the hyperplane at infinity.

\subsection{Algebraic-geometry facts}
We use elementary facts from algebraic-geometry, for which \cite{clo, shafarevich} are good references.
We do not assume that our varieties (or zerosets) are irreducible.
We will use the {\em Noether normalization lemma}.
The following statement is useful, as it characterizes the linear maps which are Noether normalizing.
\begin{theorem}{\em\cite[Thm.1.15]{shafarevich}}
  \label{thm:nn}
  If $X \subseteq \bP^{N}$ is a closed subvariety disjoint from an $\ell$-dimensional linear subspace $E \subseteq \bP^{N}$ then the projection $\pi: X \to \bP^{N-\ell-1}$ with centre $E$ defines a {\em finite} map $X \to \pi(X)$.
\end{theorem}
Here, by {\em projection with center $E$} we mean that the coordinate functions of the map are the same as a set of defining linear equations for $E$.
By the above theorem, proving that a given map is Noether normalizing for a particular variety reduces to proving that the variety is disjoint from a linear subspace.

We will also use the following two statements from dimension theory, namely the theorem on the dimension of intersections with hypersurfaces, and the theorem on the dimension of fibres.
\begin{theorem}{\em\cite[Thm.1.22]{shafarevich}}
  \label{thm:intersection}
  If a form $F$ is not zero on an irreducible projective variety $X$ then $\dim (X \cap V(F)) = \dim X - 1$.
\end{theorem}

\begin{theorem}[Fibre dimension]{\em\cite[Thm.1.25]{shafarevich}}
  \label{thm:fibredim}
  Let $f: X \to Y$ be a {\em surjective} regular map between irreducible varieties.
  Then $\dim Y \leq \dim X$, and for every $y \in Y$, the fibre $f^{-1}(y)$ satisfies $\dim f^{-1}(y) \geq \dim X - \dim Y$ (equiv.~{\em codim} $f^{-1}(y)\le\dim Y$).
  
  Further, there is a nonempty open subset $U \subset Y$: for every $y \in U$, $\dim f^{-1}(y)$ $= \dim X - \dim Y$ (equiv.~{\em codim} $f^{-1}(y)=\dim Y$).
\end{theorem}
The above theorem also holds if we replace surjective by {\em dominant}.
Every fibre either is empty (if the point is not in the image) or has the above bound on the dimension.
\ifdefined\APPENDIX
  We sketch a proof of a special case of the above in \cref{appendix:fibredim} since we require an intermediate statement in the proof of \cref{thm:trdegcompute}.
\else
  We sketch a proof of a special case of the above in Appendix A since we require an intermediate statement in the proof of \cref{thm:trdegcompute}.
\fi

We will also require the B\'ezout inequality.
The definition of degree we use is the version more common in computational complexity.
The degree of a variety is the sum of the degrees of all its irreducible components, as opposed to just the components of highest dimension.
For irreducible varieties, the degree is the number of points when intersected with a general linear subspace of complementary dimension.
This definition affords the following version of the {\em B\'ezout inequality} \cite{heintzfastqe}, which holds without any conditions on the type of intersection.
\begin{theorem}[B\'ezout \cite{heintzfastqe}]
  \label{thm:bez}
  Let $X, Y$ be subvarieties of $\bA^{n}$.
  Then $\deg (X \cap Y) \le \deg X \cdot \deg Y$.
\end{theorem}

Following is a recent version of {\em effective Nullstellensatz} \cite{jelonek}.
\begin{theorem}{\em\cite[Thm.1.1]{jelonek}}
  \label{thm:jelonekeffective}
  Let $f_{1}, \dots, f_{m}$ be nonconstant polynomials, from the ring $k\bs{x_{1}, \dots, x_{n}}$ with $k$ algebraically closed, that have no common zeros.
  Assume $\deg f_{i} = d_{i}$ with $d_{1} \geq \cdots \geq d_{m}$, and also $m \leq n$.
  Then, there exist polynomials $h_{i}$ such that $\deg f_{i} h_{i} \leq \prod_{i=1}^{m} d_{i}$ satisfying $\sum f_{i} h_{i} = 1$.
\end{theorem}

We will need the following algorithm for checking if a variety has dimension $0$ ($\dim$ is an integer in the range $[-1,n]$).
The statement assumes that the polynomials are given in the monomial (also called {\em dense}) representation.
We only state the part of the theorem that we require.
\ifdefined\APPENDIX
  A discussion is provided in Appendix \ref{appendix-0-dim}.
\else
  A discussion is provided in Appendix B.
\fi
We note that the below theorem itself invokes results from \cite{lazard1981}, section $8$ of which proves that the operations occur in a field extension of degree at most $d^{n}$ of the field $k_{i}$.
\begin{theorem}{\em\cite[Part of Thm.1]{lakshmanlazard}}
  \label{thm:ll}
  Let $f_{1}, \dots, f_{m}$ be polynomials of degree at most $d$ in $n$ variables.
  There exists a randomized algorithm that checks if the dimension of the zeroset of $f_{1}, \dots, f_{m}$ is $0$ or not, in time polynomial in $d^{n}, m$.
  The error-probability is $2^{-d^{n}}$.
\end{theorem}

We will also require a bound on the degrees of {\em annihilators} of algebraically dependent polynomials.
We refer to this bound as the Perron bound. It also plays a crucial role in the new proofs of effective Nullstellensatz (\cref{thm:jelonekeffective}).

\begin{theorem}[Perron bound]{\em\cite[Cor.5]{bms}}
  Let $f_{1}, \dots, f_{m}$ be algebraically dependent polynomials of degrees $d_{1}, \dots, d_{m}$.
  Then there exists a nonzero polynomial $A(y_{1},\dots, y_{m})$ of degree at most $\prod_{i=1}^{m} d_{i}$ such that $A(f_{1}, \dots, f_{m})$ is identically zero.
\end{theorem}
We note that the theorem statement in \cite{bms} has the bound as $(\max d_{i})^{m}$, however their method of constructing a linear faithful homomorphism and then applying the bound from \cite{ploski} actually gives the above mentioned bound (even for the {\em weighted}-degree of $A$).

In the course of our proof, we will study the image of the polynomial map whose coordinate functions are $f_{1}, \dots, f_{m}$.
We list some properties of this image.
\begin{lemma}[Polynomial map]
  \label{lem:properties}
  Let $f_{1}, \dots, f_{m}$ be polynomials of degrees at most $d$, in variables $x_{1}, \dots, x_{n}$.
  Set $r:= \trdeg{f_{1}, \dots, f_{m}}$.
  Let $F: \bA^{n} \to \bA^{m}$ be a polynomial map defined as

  \begin{center}
    $F(a_{1}, \dots, a_{n}) = (f_{1}(a_{1}, \dots, a_{n}), \dotsc, f_{m}(a_{1}, \dots, a_{n}))$.
  \end{center}

  Let $Y$ be the (Zariski) closure of the image of $\bA^{n}$ under $F$, that is $Y := \overline{F(\bA^{n})}$.
  Then,
  \begin{enumerate}[noitemsep]
    \item $Y$ is irreducible.
    \item $\dim Y = r$.
    \item $\deg Y \leq d^{r}$.
  \end{enumerate}
\end{lemma}

\begin{proof}[Proof of Lemma \ref{lem:properties}]
  The first statement is a consequence of the fact that $Y$ is the image of an irreducible set (namely $\bA^{n}$) under a continuous map.
  Since $k\bs{Y} = k\bs{f_{1}, \dots, f_{m}}$, we have $\trdeg{k\br{Y}} = r$, whence $\dim Y = r$ by definition.
  Here we used the fact that the dimension of an irreducible variety is the transcendence degree of its function field over the ground field.
  A proof of the third part can be found in \cite[8.48]{actbook}.
\end{proof}

\section{Main Results}

We require a bound on the probability that a random linear hyperplane intersects a variety of a given dimension properly, that is such that the dimension of the variety decreases by exactly one.
It is well known that the set of such hyperplanes form a Zariski open set in the space of all hyperplanes.
We use an explicit bound on the probability of such an intersection based on the degree of the variety, both for the projective and the affine case.
We will require that our intersecting hyperplanes have some structure: that their defining equations depend only on a few variables, depending on the dimension of the variety to be intersected.
We establish all these facts in the next subsection.
In the three subsections following that, we use this lemma to prove our three main results-- \cref{thm:main}, \cref{thm:main2}, and \cref{thm:trdegcompute}.

\subsection{Intersection by a hyperplane}

\begin{lemma}
  \label{lem:intdim}
  Let $V \subseteq \bP^{n}$ be a projective variety of dimension $r$ and degree $D$.
  Let $S$ be a finite subset, of the underlying field $k$, not containing $0$.
  Let $\ell$ be a linear form in $x_{0}, x_{1}, \dots, x_{n-r}$ with each coefficient picked uniformly and independently from $S$.
  Let $H$ be the hyperplane defined by $\ell$.
  Then, with probability at least $1 - D/\abs{S}$ we have $\dim V \cap H = \dim V - 1$.

  Analogously, if $V\subseteq \bA^{n}$ is affine, $\ell$ is a linear polynomial in $x_{1}, \dots, x_{n-r+1}$ and $H$ its hyperplane; then $\dim V \cap H = \dim V - 1$ with probability at least $1 - 2D/\abs{S}$.
\end{lemma}
\begin{proof}[Proof of Lemma \ref{lem:intdim}]
  First we prove the projective case.
  Let $\ell := c_{0}x_{0} + \cdots + c_{n-r} x_{n-r}$, where the $c_{i}$ are the coefficients picked uniformly at random from $S$.
  Let $\cup_{j=1}^{d} V_{j}$ be the decomposition of the dimension-$r$ part of $V$ into irreducible components.
  Then by definition, $\deg V \geq \sum \deg V_{j}$, and hence $d \leq D$.
  Pick a point $p_{j}$ in $V_{j}$, for each $j$.
  We can always pick $p_{j}$ so that not all of its first $n-r+1$ coordinates are zero: if this was not possible then $V_{j}$ would have to be contained in the variety defined by $x_{0} = x_{1} = \cdots = x_{n-r} = 0$, which has dimension only $r-1$.
  By \cref{thm:intersection}, $\dim H \cap \dim V_{j} = \dim V_{j}$ if and only if $V_{j} \subseteq H$ (since $V_{j}$ and $H$ are irreducible),  and otherwise $\dim H \cap V_{j} = \dim V_{j} - 1$.
  The probability that this happens is upper bounded by the probability that $p_{j} \in H$.
  For a fixed $j$, this is equivalent to $\ell(p_{j}) = 0$.
  Since not all of the first $n-r+1$ coordinates of $p_{j}$ are zero, the above is bounded by $1/\abs{S}$, by fixing all but one of the coordinates.
  By a union bound, with probability at most $d/\abs{S}$, there exists some $j$ where $\dim H \cap V_{j} = \dim V_{j}$.
  Therefore, with probability at least $1-D/\abs{S}$, we get $\dim V_{j} \cap H = \dim V_{j} - 1$ for every $j$, whence $\dim V \cap H = \dim V - 1$.

  Now suppose $V$ is affine.
  The difference from the projective case is that the intersection $V \cap H$ might be empty, and we need to bound the probability of this event.
  Let $V^{p}$ be its projective closure.
  Then $\dim V^{p} = \dim V$ and $\deg V^{p} = \deg V$.
  By the previous part, we have $\dim V^{p} \cap H^{p} = \dim V^{p} - 1$ with probability $1 - D/\abs{S}$.
  Then, the case $V \cap H = \emptyset$ only happens if $\dim V^{p} \cap H^{p} \cap \bP_{\infty}^{n} = \dim V^{p} - 1$, where $\bP_{\infty}^{n}$ is the hyperplane $x_{0} = 0$ in $\bP^{n}$.
  The irreducible components of $V$ are in bijection with those of $V^{p}$, and hence $V^{p}$ has no irreducible component contained in $\bP^{n}_{\infty}$.
  Therefore, $\dim V^{p} \cap \bP^{n}_{\infty} = \dim V^{p} - 1$.
  Further, by B\'ezout's theorem we have $\deg V^{p} \cap \bP^{n}_{\infty} \leq \deg V^{p}$.
  
  Now $H^{p} \cap \bP^{n}_{\infty}$ is a hyperplane in $\bP^{n}_{\infty}$ defined by the nonconstant part of $\ell$.
  In particular, it is a hyperplane whose defining equation has coefficients picked uniformly and independently and we can apply the projective version of this lemma on $\bP^{n}_{\infty}$.
  Therefore the probability that its intersection with $V^{p} \cap \bP^{n}_{\infty}$ does not result in a reduction in the dimension is at most $D/\abs{S}$.
  By a union bound, with probability at least $1-2D/\abs{S}$ it holds that $\dim V^{p} \cap H^{p} = \dim V - 1$ and $\dim V^{p} \cap H^{p} \cap \bP^{n}_{\infty} = \dim V - 2$, whence $\dim V \cap H = \dim V - 1$ as required.
\end{proof}
An important fact to note is that our choice of variables for the linear form is arbitrary.
The lemma works for any choice of $n-r+1$ variables, and this will be important when we use the lemma.
Also, note that the above lemma works when our linear form involves more that $n-r+1$ variables.

Repeated applications of the above allow us: (1) to reduce a variety to dimension $0$ by taking hyperplane sections, and (2) to find a linear subspace that avoids the variety.

\subsection{Radical membership: Proof of {\cref{thm:main}}}\label{sec-rad-memb}
Using the above lemma, we complete the proof of the main theorem:
%
\begin{proof}[Proof of \cref{thm:main}]
  We first assume $g = 1$, which is the Nullstellensatz problem HN.
  Define $D:= \prod_{i=1}^{m} d_{i}$, and $V := V(\ideal{\vf})$.
  The set of common zeroes of these polynomials is the fibre of the point $\vzero$ under the map $F$ defined in \cref{lem:properties}.
  The problem HN is thus equivalent to testing if a particular fibre of a polynomial map is nonempty.
  By the fibre dimension theorem (\cref{thm:fibredim}), the codimension of the zeroset---if it is nonempty---is bounded above by the dimension of the image of the map, which by \cref{lem:properties} is $r$.
  The zeroset $V$ is therefore either empty, or has dimension at least $n-r$.
  Assume that $V$ is nonempty.
  By repeated applications of B\'ezout's theorem (\cref{thm:bez}), $\deg V \leq D$.
  Let $S$ be a subset of the underlying field $k_{i}$ (or an extension) of size at least $6(n-r)D$ that does not contain $0$.
  We can sample from $S$ in time polynomial in $d,n,m$, since $S$ has size exponential in these parameters.
  Further, if we were required to go to an extension to form $S$, the degree of the extension would be polynomial in $d, n, m$.
  Pick $n-r$ random linear polynomials $\ell_{1}, \dots, \ell_{n-r}$ with coefficients from $S$, and call their zero sets $H_{1}, \dots, H_{n-r}$ respectively.
  By Lemma \ref{lem:intdim}, the intersection $V \cap H_{1}$ has dimension $r-1$ with probability at least $1-1/(3(n-r))$.
  Further, by Bézout's theorem we get $\deg V \cap H_{1} \leq \deg V \leq D$, since each $H_{i}$ has degree one.
  Again by Lemma \ref{lem:intdim}, the intersection $(V \cap H_{1}) \cap H_{2}$ has dimension $r-2$ with probability at least $1-1/(3(n-r))$, and $\deg V \cap H_{1} \cap H_{2} \leq D$.
  Repeating this for all $H_{i}$ and using the union bound, we get $\dim V \cap H_{1} \cap \cdots \cap H_{n-r} \geq 0$ with probability at least $2/3$.

  Therefore, when the polynomials $\vf$ have nonempty zeroset and are restricted to the $r$ dimensional affine subspace $\cap H_{i}$, the new zeroset has dimension at least $0$, and in particular is nonempty.
  If the zeroset of the polynomials was empty to begin with, then the restriction to the linear subspace also results in an empty zeroset.

  This restriction can be performed by a variable reduction, as follows.
  Treating $\bA^{n}$ as a vector space of dimension $n$ over $k$, let $H_{0}$ be the linear subspace corresponding to the affine subspace $H:=\cap H_{i}$.
  $H_0$ has dimension $r$, and hence has basis $a_{1}, \dots, a_{r}$.
  Further, let vector $b$ be such that $H = H_{0} + b$.
  Define linear forms $c_{1}, \dots, c_{n}$ in new variables $z_{1}, \dots, z_{r}$ as $c_{i} := \sum_{j=1}^{r} a_{ji} z_{j} + b_i$, where $a_{ji}$ is the $i^{th}$ component of $a_{j}$.
  Define $f'_{i} := f_{i}(c_{1}, \dots, c_{n})$.
  Then by construction, the zeroset of $f'_{1}, \dots, f'_{m}$ is equal to $V \cap (\cap H_{i})$.
  Further, $\deg f'_{i} = \deg f_{i}$, and these polynomials are in $r$ variables.
  Also, the construction of these $f'_{i}$ can be done in a blackbox manner, given blackboxes for $f_{i}$.
  This construction takes time polynomial in $m, r, n$.

  We now repeatedly invoke \cref{thm:ll} to check if $f'_{i}$s have a common root.
  First we must convert them to a sparse representation.
  The polynomial $f'_{i}$ has at most $\binom{r+d_{i}}{r}$ many monomials, and therefore we can find every coefficient in time polynomial in $\binom{r+d_{i}}{r}$ by simply solving a linear system.
  Applying \cref{thm:ll}, we can test whether the dimension of the zeroset of $f'_{1}, \dots, f'_{m}$ is $0$ or not.
  However, we want to check if the dimension is at least $0$.
  For this, we randomly sample $r$ more hyperplanes $H'_{1}, \dotsc, H'_{r}$ as in the previous part of the proof, this time in the new variables $z_{1}, \dots, z_{r}$.
  Let $V'$ be the zeroset of $f'_{1}, \dots, f'_{m}$.
  We first use \cref{thm:ll} to check if $V'$ has dimension $0$.
  If not, then we check if $V' \cap H'_{1}$ has dimension $0$.
  If not, then we check $V' \cap H'_{1} \cap H'_{2}$, and so on.
  We return success if any one of the above iterations returns success (implying that the corresponding variety has dimension $0$).
  Performing calculations similar to the ones earlier in the proof, we see that with high probability each intersection reduces the dimension by $1$.
  If $V'$ originally had dimension $r'$, then after intersecting with $r'$ hyperplanes, the algorithm of \cref{thm:ll} returns success.
  If $V'$ was empty, then the algorithm does not return success in any of the above iterations.
  This allows us to decide if $V'$ has dimension at least $0$.
  Finally, using the fact that the dimension of the zeroset of $f'_{1}, \dots, f'_{m}$ is at least $0$ if and only if $\dim V\ge0$, we get the required algorithm for HN.

  We now estimate the time taken.
  Computing the dense representation takes time polynomial in $d^{r}$ and $m$.
  Each of the at most $r$ applications of \cref{thm:ll} also take the same amount of time.
  The sampling steps take time polynomial in $\log nD$ (in turn polynomial in $d,m$) and only requires an extension of degree polynomial in $n$ and $\log{d}$.
  The total time taken is therefore polynomial in $m,d^{r}$.

  Now assume that $g$ is an arbitrary polynomial.
  We reduce the problem to the case of $g=1$ using Rabinowitsch trick \cite{rabinowitsch1930hilbertschen}.
  The polynomial $g$ belongs to the radical of the ideal $\ideal{\vf}$ if and only if the polynomials $\vf, 1-yg$ have no common root (here $y$ is a new variable).
  Further, if $\vf$ have transcendence degree $r$, then the set $\vf, 1-yg$ has transcendence degree $r+1$.
  We therefore reduce the radical membership problem to HN problem, with a constant increase in the transcendence degree, number of polynomials and the number of variables.
  By the result in the previous paragraph, we can solve this in time polynomial in $n, m$ and $d^{r}$.
\end{proof}

\subsection{Effective Nullstellensatz: Proof of {\cref{thm:main2}}}
We now prove that by taking random linear combinations of the input polynomials, we can reduce the number of polynomials to be one more than the transcendence degree while preserving the existence of roots.
This reduction gives degree bounds for the Nullstellensatz certificates.
Note that this reduction does not help in Section \ref{sec-rad-memb}'s root-testing procedure, since we will only be saving a factor in $m$ if we reduce the number of polynomials.
\begin{theorem}[Generator reduction]
  \label{thm:polynomialreduction}
  Let $f_{1}, \dots, f_{m}$ be polynomials, in $x_{1}, \dots, x_{n}$, of degrees atmost $d$ and of transcendence degree $r$.
  Let $g_{1}, \dots, g_{r+1}$ be polynomials defined as $g_{i} := \sum_{j=i}^{m} c_{ij} f_{j}$, where each $c_{ij}$ is randomly picked from a finite subset $S$ of $k$.
  Then with probability at least $1-d^{(r+1)m}/\abs{S}$, we have $V(\ideal{\vf}) = V(\ideal{\vg})$.
\end{theorem}
That we pick the linear combinations so that the first involves all polynomials, the second involves all except $f_{1}$, the third involves all except $f_{1}, f_{2}$ and so on is crucial for the improvement in the degree bounds.
\begin{proof}[Proof of \cref{thm:polynomialreduction}]
  We prove this by studying the set $Y$ defined in \cref{lem:properties}.
  Let $F: \bA^{n} \to \bA^{m}$ be the map with coordinate functions $f_{i}$.
  Let $Y := \overline{F(\bA^{n})}$, the closure of the image of $F$ in $\bA^{m}$.
  We use $y_{1}, \dots, y_{m}$ to denote the coordinate functions of $\bA^{m}$.
  By \cref{lem:properties}, $Y$ has dimension $r$, and degree at most $D := d^{r}$.
  Let $Y^{p}$ be the projective closure of $Y$.
  Then $Y^{p}$ also has dimension $r$ and degree at most $D$.
  Let $\ell_{1}, \dots, \ell_{r+1}$ be the linear polynomials $\ell_{i} := \sum_{i\le j\le m} c_{ij} y_{j}$.

  Consider the subspace defined by $y_{0}, \ell_{1}, \dots, \ell_{r}$ in $\bP^{m}$.
  The variety $Y^{p} \cap \bP^{m}_{\infty}$, which is the intersection of $Y^{p}$ with the hyperplane at infinity defined by $y_{0} = 0$, has dimension $r-1$.
  Since $\ell_{1}, \dots, \ell_{r}$ are random linear polynomials and $Y^{p} \cap \bP^{m}_{\infty}$ is a variety of, degree at most $D$ and, dimension $r-1$, we can repeatedly apply Lemma \ref{lem:intdim} to get a bound on the probability of proper intersections.
  Let $H_{i}$ be the hyperplane defined by $\ell_{i}$.
  We apply Lemma \ref{lem:intdim} starting from $H_{r}$.
  The equation $\ell_{r}$ has $m-r+1$ coefficients, and therefore satisfies the conditions required for the lemma.
  By B\'ezout's theorem, the intersection has degree bounded by $D$, and dimension decreased by one.
  We then apply the theorem with $H_{r-1}$ and so on, as in the proof of \cref{thm:main}.
  In each iteration the variety considered has one less dimension than the previous iteration, but our linear polynomial has one more variable, and therefore we will always satisfy the conditions of Lemma \ref{lem:intdim}.

  We can now invoke \cref{thm:nn} to say that the map $\bP^{m} \to \bP^{r}$ with coordinate functions $(y_{0}, \ell_{1}, \dots, \ell_{r})$ is Noether normalizing for $Y^{p}$.
  We call this map $L'$.
  We use $z_{0}, \dots, z_{r}$ to denote the coordinate functions of $\bP^{r}$.
  The map $L'$ sends the affine chart $y_{0} \neq 0$ to the affine chart $z_{0} \neq 0$.
  Let $L$ be the restriction of $L'$ to this affine chart.
  Then $L$ defines a map from $\bA^{m}$ to $\bA^{r}$, which is Noether normalizing for the variety $Y$; we also call this restricted map $L$.
  More explicitly, the map $L$ has coordinate functions $(\ell_{1}, \dots, \ell_{r})$.
  Also, let the map $\bA^{m} \to \bA^{r+1}$ with coordinate functions $(\ell_{1}, \dots, \ell_{r+1})$ be labelled $M$.

  Since the map $L$ is Noether normalizing, it has finite fibres.
  Let $Q$ be the fibre of $\vzero$ in $Y$.
  We bound the size of this set.
  The map $L$ is Noether normalizing, and hence it is surjective.
  The image $\bA^{r}$ is normal, and hence the cardinality $|Q|$ of the fibre is bounded by the degree of the map \cite[Theorem~2.28]{shafarevich}.
  Here, by the degree of the map we mean the degree of $k\br{Y}$ over the pullback $L^{*}(k\br{\bA^{r}})$.
  Note that $k\br{Y}=k(f_{1}, \dots, f_{m})$, and $L^{*}(k\br{\bA^{r}})=k(\ell_{1}(\vf), \dots, \ell_{r}(\vf))$ after applying the same isomorphism.
  By Perron's bound, for each $i$ there exists an annihilator of $f_{i}, l_{1}(\vf), \dots, l_{r}(\vf)$ of degree at most $d^{r+1}$.
  The degree of the extension, and hence $\abs{Q}$, is bounded by $d^{m(r+1)}$.
  
  Further, no point of $Q$, other than $\vzero$, has all of the last $m-r$ coordinates as zero.
  This follows from the fact that $L^{-1}(\vzero)$ is a linear space of dimension $m-r$, and its intersection with $y_{r+1} = y_{r+2} = \cdots = y_{m} = 0$ has dimension $0$.
  Consider now the linear form $\ell_{r+1}$.
  For every $\vzero\ne q \in Q$, the probability that $\ell_{r+1}(q) = 0$ is at most $1/\abs{S}$.
  Therefore, with probability at least $1-d^{m(r+1)} / \abs{S}$, the polynomial $\ell_{r+1}$ is nonzero on every nonzero point of $Q$.  

  Consider the polynomials $g_{1}, \dots, g_{r+1}$, and let $G$ be the polynomial map $\bA^{n} \to \bA^{r+1}$ with coordinate functions $g_{i}$.
  By the choice of $\ell_{i}$ in the previous paragraph, the map $G$ is exactly the composition of the map $F:\bA^{n} \to \bA^{m}$ with $M:\bA^{m} \to \bA^{r+1}$.
  Let $Q$ be as defined earlier, the fibre of $\vzero$ under $L$.
  By construction, the set $M^{-1}(\vzero)$ is a subset of $Q$.
  But since the polynomial $\ell_{r+1}$ is nonzero on every nonzero point of $Q$, the set $M^{-1}(\vzero)$ consists only of $\vzero$.
  Therefore, $F^{-1}(M^{-1}(\vzero)) = F^{-1}(\vzero)$.
  Since $G = M \circ F$ we get $G^{-1}(\vzero) = F^{-1}(\vzero)$; which is the same as $V(\ideal{\vf}) = V(\ideal{\vg})$.
\end{proof}
We use the above to prove our 2nd main result:
%
\begin{proof}[Proof of \cref{thm:main2}]
  Using \cref{thm:polynomialreduction}, there exists polynomials $g_{1}, \dots, g_{r+1}$ of degrees $d_{1}, \dots, d_{r+1}$ that do not have a common root.
  By \cref{thm:jelonekeffective}, there exist $h'_{1}, \dots, h'_{r+1}$ such that $\deg g_{i} h'_{i} \leq \prod_{i=1}^{r+1} d_{i}$ such that $\sum g_{i} h'_{i} = 1$.
  In this equation, substitute back the linear-combination of $f_{1}, \dots, f_{m}$ for each $g_{i}$; whence we get the required $h_{i}$'s.
\end{proof}

\subsection{Computing transcendence degree: Proof of {\cref{thm:trdegcompute}}}
We give a method of `efficiently' computing the transcendence degree of input polynomials $f_{1}, \dots, f_{m}$.
By \cref{lem:properties} and the second part of \cref{thm:fibredim}, the transcendence degree can be computed if we know the dimension of a general fibre.
We need to get a bound on the points that violate the equality in \cref{thm:fibredim}.
For this we follow the classical proof of the theorem and give effective bounds wherever required.
\ifdefined\APPENDIX
  For convenience we have provided a proof sketch in \cref{appendix:fibredim}, for the special case we need.
\else
  For convenience we have provided a proof sketch in Appendix A, for the special case we need.
\fi

\begin{lemma}
  \label{lem:effectivedim}
  Let $h_{1}, \dots, h_{m}$ be polynomials of degree at most $d$ in $n$ variables, and let $W$ be the Zariski closure of the image of the map $\vh$ with coordinates $h_{i}$.
  Let $S\subset k$ be of size $6nd^{n}$.
  If $a_{1}, \dots, a_{n}$ are randomly picked from $S$, then with probability at least $5/6$, the fibre of $(h_{1}(\va), \cdots, h_{m}(\va))$ has codimension exactly $\dim W$.
\end{lemma}
\begin{proof}
  First assume that the $h_{i}$ are algebraically {\em independent}.
  Then $W = \bA^{m}$.
  Let the input variables be labelled such that $x_{1}, \dots,$ $x_{n-m}, h_{1}, \dots, h_{m}$ are algebraically independent, and let $A_{j}(z_{0}$, $z_{1}, \dots,$ $z_{n-m}, w_{1}, \dots, w_{m})$ be the (minimal) annihilator of $x_{j}$ over this set of variables, that is $A_{j}(x_{j}, x_{1}, \dots, x_{n-m}, h_{1}, \dots, h_{m}) = 0$.
  \ifdefined\APPENDIX
    By the proof of \cref{thm:fibredim} (\cref{appendix:fibredim}), a sufficient condition for point $a_{1}, \dots, a_{n}$ to be such that $\vh(\va)$ has fibre of dimension exactly $n-m$ is that $A_{j}(x_{j}, x_{1}, \dots, x_{n-m}, h_{1}(\va), \dots, h_{m}(\va))$ is a nonzero polynomial.
  \else
    By the proof of \cref{thm:fibredim} (Appendix A), a sufficient condition for point $a_{1}, \dots, a_{n}$ to be such that $\vh(\va)$ has fibre of dimension exactly $n-m$ is that $A_{j}(x_{j}, x_{1}, \dots, x_{n-m}, h_{1}(\va), \dots, h_{m}(\va))$ is a nonzero polynomial.
  \fi
  The polynomial $A_{j}$, when treated as polynomials in variables $z_0, \dots, z_{n-m}$ with coefficients in $k\bs{w_{1}, \dots, w_{m}}$ are such that the leading monomial has coefficient a polynomial in $w_{1}, \dots, w_{m}$ of weighted-degree at most $\prod_{i=1}^{m} d_{i}$ (by Perron bound).
  By the polynomial identity lemma \cite{ore, demillolipton, schwartz, zippel}, if we pick each $a_{i}$ randomly from a set of size $6\prod_{i=1}^{m} d_{i}$ then, with probability at least $5/6$, none of the polynomials $A_{j}(x_{j}, x_{1}, \dots, x_{n-m}$, $h_{1}(\va), \dots, h_{m}(\va))$ is zero.
  In this case, the codimension of the fibre of $\vh(\va)$ is exactly $m$ as claimed.

  In the general case, the $h_{i}$ may be algebraically dependent, and $W$ is a subvariety of $\bA^{m}$.
  Suppose $\dim W= \trdeg{\vh} =: s$.
  Then we take $s$ many random linear combinations $g_{i}$ of the $h_{i}$, as in the proof of \cref{thm:main2}.
  The map defined by the $g_{i}$ is dense in $\bA^{s}$ and therefore the $g_{i}$ ($i\in[s]$) are algebraically independent.
  By the previous paragraph, point $\va$ picked coordinatewise from $S$ is such that the fibre of $\vg(\va)$ has codimension $s$.
  The fibre of $\vh(\va)$ is a subset of the fibre of $\vg(\va)$, and therefore it has codimension at least $s$.
  Finally, by \cref{thm:fibredim}, the fibre has codimension at most $s$, whence the fibre of $\vh(\va)$ has codim $=s$ as required.
\end{proof}

\begin{proof}[Proof of \cref{thm:trdegcompute}]
  For each $i$, upwards from $1$ to $n$, we do the following steps.
  We iterate till $i$ reaches transcendence degree $r$ of the $m$ polynomials.
  In the $i$-th iteration, we intersect $\bA^{n}$ with $n-i$ random hyperplanes $\ell_1,\dots,\ell_{n-i}$, as in the proof of \cref{thm:main} (Sec.\ref{sec-rad-memb}).
  Here, the coefficients are picked from a set $S$ of size at least $n\cdot18\prod_{i=1}^{m} d_{i}$.
  We therefore reduce the problem to $i$ variables.

  Randomly pick point $\va$ where each coordinate (of the $n$ many) is picked randomly from $S$. By Lemma \ref{lem:effectivedim} (\& \ref{lem:intdim}), with error-probability $\le 1/6n$, the point $\vf(\va)$ has intersected fibre of dimension $(n-r)-(n-i)=$ $(i-r)$. We need to check this algorithmically; which is done by interpolating the polynomials $\vf$ after hyperplane intersections, and then using \cref{thm:ll} (as detailed in Sec.\ref{sec-rad-memb}).
  If the intersected fibre dimension is zero, we have certified transcendence degree $=i=r$; so we halt and return $i$ as output.
  Else, we move to the next $i\mapsto i+1$. 
  The interpolation step above is performed by solving a linear system which has size polynomial in $d^{i}$ which is the count of the monomials of degree at most $d$ in $i$ variables.

  Note that for $i<r$, with error-probability $\le 1/6n$, the fibre of $\vf(\va)$ has an empty intersection with $\ell_1,\dots,\ell_{n-i}$; which is dim$=-1$ and hence gets verified by \cref{thm:ll}. 
  
  By a union bound therefore, with error-probability $\le 1/6$, the above algorithm gives the correct answer.
  For each $i$, the time complexity of the above steps is polynomial in $d^{i},m$, which is the time taken for the interpolation step and to verify zero-dimension of the fibre.
  Therefore the algorithm as a whole takes time polynomial in $d^{r},n,m$ as claimed.
\end{proof}

\section{Conclusion}
We give algorithms for radical membership and transcendence degree of systems of polynomials, in time that depends on the transcendence degree.
In both cases, our algorithms generalize the cases of `few' input polynomials.
We further give bounds on the degree of the Nullstellensatz certificates that depend on the transcendence degree of the input polynomials.
In all three cases, our bounds are significantly better than the previously known results in the regime when the transcendence degree is much smaller than the number of variables and the number of polynomials.

Our work leaves the natural open problem of designing efficient algorithms when the transcendence degree is `larger'.
\begin{itemize}
\item For the blackbox radical membership problem, given the NP-hardness of HN, it is unlikely that a significantly better algorithm exists (unless other restrictions are put on the input polynomials). 

\item Could our methods, and the core hard instance thus identified, help in proving that HN is in AM? Currently, this is known only partially \cite{koiranph}.

\item For the transcendence degree problem however, we know that the problem is in coAM $\cap$ AM, making it unlikely to be NP hard.
It is therefore likely that there is an efficient randomized algorithm whose time complexity is polynomial in $n$ and $m$.
This is already known in the case when the field has large/zero characteristic, and it is an open problem to extend this to other fields.
A first step might be to give a subexponential time algorithm for the problem that works without any assumptions.
\end{itemize}

\noindent
{\bf Acknowledgements. } 
We thank Ramprasad Saptharishi for introducing us to the universality of 3-generators ideal membership problem. Nitin Saxena~thanks the funding support from DST (DST/SJF/MSA-01/2013-14).

\bibliographystyle{alpha}
\bibliography{ref}

\ifdefined\APPENDIX
  \appendix
\section{Effective proof of fibre dimension theorem for a dominant map}
\label{appendix:fibredim}
In this section, we sketch a proof of \cref{thm:fibredim} and show that the condition on $A_{j}$ in the proof of \cref{thm:trdegcompute} is indeed sufficient.
We follow the proof in \cite{shafarevich}.
We only sketch the proof in the case when $Y = \bA^{M}$ and $X = \bA^{N}$; however we do not need `surjectivity' and only assume a dominant $f$.

Let $(b_{1}, \dots, b_{M}) \in \bA^{M}$ be a point in the image of $f$.
This point is defined by the equations $y_{1} - b_{1}, \dots, y_{M}-b_{M}$.
If $f_{1}, \dots, f_{M}$ are the coordinate functions of $f$, then the fibre $f^{-1}(\vb)$ is defined by the equations $f_{i} - b$ in $\bA^{N}$.
The nonempty fibre is defined therefore by $M$ equations; thus by \cref{thm:intersection}, every component has dimension at least $N-M$.

Now we prove the second part.
Since the map $f$ is dominant, the induced map $f^{*}: k\bs{Y} \to k\bs{X}$ is injective, and we identify $k\br{Y}$ with its image in $k\br{X}$ via $f^{*}$.
Then by dimension definition, $k\br{X}$ has transcendence degree $N-M$ over $k\br{Y}$.
Let $y_{1}, \dots, y_{M} \in k\br{Y}$ be the generators of $k\br{Y}$ over $k$ and let $x_{1}, \dots, x_{N}$ be the generators of $k\br{X}$ over $k$.
Let them be numbered so that $x_{1}, \dots, x_{N-M}, y_{1}, \dots, y_{M}$ are algebraically independent, and let $A_{j}$ be the annihilator of $x_{j}$ over this set.
The fibre $f^{-1}(\vb)$ of a point has coordinate ring generated by the restrictions of $x_{1}, \dots, x_{N}$ to the fibre.
Suppose the point $\vb$ is such that the annihilators $A_{j}$ are nonzero polynomials when we substitute $\vb$ for $\vy$.
Then $A_{j}$ continue to be annihilators of the restrictions of $x_{j}$ over the restrictions of $x_{1}, \dots, x_{N-M}$.
This shows that the fibre has dimension at most $N-M$, and combining this with the previous part, the fibres have dimension exactly $N-M$.

Therefore, the sufficient condition for the fibres to have this dimension is that all the $A_{j}$ remain nonzero polynomials when evaluated at the given point $\vb$.
This completes our proof sketch.

\section{Tools for zero-dimension testing}\label{appendix-0-dim}
We briefly discuss \cref{thm:ll}.
The two main tools used here are Lazard's algorithm \cite{lazard1981} for computing a multiple of the U-resultant \cite{macaulay1902}, and Canny's \cite{canny} study of this U-resultant.

Given input polynomials $f_{1}, \dots, f_{m}$, first the polynomials are reduced to $n+1$ many polynomials, by taking random linear combinations.
This does not change the root set, and we assume $n = m$.
The polynomials are then {\em deformed} and homogenized: set $h_{i}$ to be the homogenization of the polynomial $f_{i} + sx_{i}^{\deg f_{i}}$.
Here $s$ is a new variable.
This is due to \cite{canny}, and is useful because it gives control over the newly introduced roots at the hyperplane at infinity due to homogenization.

The next step is to use the algorithm of \cite{lazard1981} to compute matrices $M_{0}, \dots, M_{n}$ (with each entry a polynomial in $s$) that are such that the determinant of $\sum u_{i} M_{i}$ is the U-resultant up to a multiple.
To compute these matrices, first the matrices of multiplication by $x_{i}$ considered as a map between the homogeneous degree $D-1$ and degree $D:= \sum d_{i} - n$ components of the coordinate ring of the zeroset is considered.
A simultaneous base change is performed on these matrices, and then a subset of the columns of each matrix is returned as $M_{0}, \dots, M_{n}$.
The computation of these matrices is the most expensive step in the algorithm.
Since the number of $n$-variate monomials of degree $D$ is $d^{O(n)}$, the algorithm takes time just polynomial in $d^{n},m$.

The coefficient of the lowest degree term (in $s$) in the above determinant is a product of linear forms (in $\mathbf{u}$), whose coefficient correspond to the isolated zeros (that is, zeroes not part of a higher dimension component) of the original polynomial \cite{canny}.
However, computing this coefficient is prohibitively expensive, and therefore some specializations of the U-resultant are computed.
It is at this specialization stage that the algorithm decides the zero dimensionality.
First, a change of basis is done to ensure that all the roots of the homogenized polynomials have distinct first coordinates.
The variables $U_{0}, \dots, U_{n}$ are specialized to $(y, 1, 0, \dots, 0)$ where $y$ is a new variable.
The determinant of the matrix $\sum U_{i} M_{i}$ is now a polynomial in $y$ and $s$, and the coefficient of the lowest degree term is nonzero in $y$ if and only if the original system of polynomials has zeroset of dimension exactly $0$.
This last test can be done just by studying the characteristic polynomial of the matrix $M_{1} M_{0}^{-1}$, and this can also be done in time polynomial in $d^{n},m$.

\else
\fi
\end{document}